\documentclass[12pt]{article}
\usepackage[utf8]{inputenc}
\usepackage{amsfonts,latexsym,amsmath,amsthm,amssymb,amscd,euscript,mathrsfs,graphicx}
\usepackage{framed}
\usepackage{fullpage}
\usepackage{color}
\usepackage{enumitem}
\usepackage[obeyFinal,textsize=scriptsize,shadow]{todonotes}
\usepackage{tikz}
\usetikzlibrary{matrix}
\usepackage[hidelinks]{hyperref}
\usepackage{cleveref}
\usepackage{soul} 
\usepackage{todonotes}
\usepackage{algorithm}
\usepackage[noend]{algpseudocode}

\theoremstyle{definition}
\newtheorem{theorem}{Theorem}[section]
\newtheorem{proposition}[theorem]{Proposition}
\newtheorem{fact}[theorem]{Fact}
\newtheorem{lemma}[theorem]{Lemma}

\theoremstyle{remark}

\newcommand{\eps}{\varepsilon}
\newcommand{\ten}{100}

\DeclareMathOperator{\Bin}{Bin}

\DeclareMathOperator{\val}{val}
\DeclareMathOperator{\rev}{rev}
\DeclareMathOperator*{\argmax}{arg\,max}

\DeclareMathOperator*{\polylog}{polylog}
\DeclareMathOperator*{\poly}{poly}

\newcommand{\fC}{\mathfrak{C}}
\newcommand{\fD}{\mathfrak{D}}

\newcommand{\ba}{\mathbf a}
\newcommand{\bb}{\mathbf b}

\title{Near-Optimal Trace Reconstruction for Mildly Separated Strings}
\author{Anders Aamand\thanks{\texttt{aa@di.ku.dk}. University of Copenhagen. This work was supported by the DFF-International Postdoc Grant 0164-00022B and by the VILLUM Foundation grants 54451 and 16582.} \and Allen Liu\thanks{\texttt{cliu568@mit.edu}. Massachusetts Institute of Technology. This work was supported in part by an NSF Graduate Research Fellowship, a Hertz Fellowship, and a Citadel GQS Fellowship.} \and Shyam Narayanan\thanks{\texttt{shyam.s.narayanan@gmail.com}.Massachusetts Institute of Technology. Supported by the NSF TRIPODS Program, an NSF Graduate Fellowship, and a Google Fellowship.}}
\date{\today}

\begin{document}
\maketitle
\begin{abstract}
    In the \emph{trace reconstruction} problem our goal is to learn an unknown string $x\in \{0,1\}^n$ given independent \emph{traces} of $x$. A trace is obtained by independently deleting each bit of $x$ with some probability $\delta$ and concatenating the remaining bits. It is a major open question whether the trace reconstruction problem can be solved with a polynomial number of traces when the deletion probability $\delta$ is constant. The best known upper bound and lower bounds are respectively $\exp(\tilde O(n^{1/5}))$~\cite{Chase21b} and $\tilde \Omega(n^{3/2})~\cite{Chase21a}$. Our main result is that if the string $x$ is \emph{mildly separated}, meaning that the number of zeros between any two ones in $x$ is at least $\polylog n$, and if $\delta$ is a sufficiently small constant, then the trace reconstruction problem can be solved with $O(n \log n)$ traces and in polynomial time. 
\end{abstract}

\section{Introduction}
\emph{Trace reconstruction} is a well-studied problem at the interface of string algorithms and learning theory. Informally, the goal of trace reconstruction is to recover an unknown string $x$ given several independent noisy copies of the string.

Formally, fix an integer $n \ge 1$ and a deletion parameter $\delta \in (0, 1)$. Let $x \in \{0, 1\}^n$ be an unknown binary string with $x_i$ representing the $i$th bit of $x$. Then, a \emph{trace} $\tilde{x}$ of $x$ is generated by deleting every bit $x_i$ independently with probability $\delta$ (and retaining it otherwise), and concatenating the retained bits together. For instance, if $x = 01001$ and we delete the second and third bits, the trace would be $001$ (from the first, fourth, and fifth bits of $x$).
For a fixed string $x$, note that the trace follows some distribution over bitstrings, where the randomness comes from which bits are deleted.
In \emph{trace reconstruction}, we assume we are given $N$ i.i.d. traces $\tilde{x}^{(1)}, \dots, \tilde{x}^{(N)}$, and our goal is to recover the
original string $x$ with high probability.

The trace reconstruction problem has been a very well studied problem over the past two decades~\cite{Levenshtein01a, Levenshtein01b, BatuKKM04, KannanM05, HolensteinMPW08, ViswanathanS08, McGregorPV14, DeOS17, NazarovP17, PeresZ17, HartungHP18, HoldenL18, HoldenPP18, Chase21a, ChenDLSS21a, ChenDLSS21b, Chase21b, Rubinstein23}.
There have also been numerous generalizations or variants of trace reconstruction studied in the literature, including coded trace reconstruction~\cite{CheraghchiGMR19, BrakensiekLS19}, reconstructing mixture models~\cite{BanCFSS19, BanCSS19, Narayanan21}, reconstructing alternatives to strings~\cite{DaviesRR19, KrishnamurthyMMP21, NarayananR21, McGregorS22, SunY23, McGregorS24}, and approximate trace reconstruction~\cite{DaviesRSR21, ChaseP21, ChakrabortyDK21, ChenDLSS22, ChenDLSS23}.

In perhaps the most well-studied version of trace reconstruction, $x$ is assumed to be an arbitrary $n$-bit string and the deletion parameter $\delta$ is assumed to be a fixed constant independent of $n$. In this case, the best known algorithm requires $e^{\tilde{O}(n^{1/5})}$ random traces to reconstruct $x$ with high probability~\cite{Chase21b}. As we do not know of any polynomial-time (or even polynomial-sample) algorithms for trace reconstruction, there have been many works making distributional assumptions on the string $x$, such as $x$ being a uniformly random string~\cite{HolensteinMPW08, McGregorPV14, PeresZ17, HoldenPP18, Rubinstein23} or $x$ being drawn from a ``smoothed'' distribution~\cite{ChenDLSS21a}.
An alternative assumption is that the string $x$ is parameterized, meaning that $x$ comes from a certain ``nice'' class of strings that may be amenable to efficient algorithms~\cite{KrishnamurthyMMP21, DaviesRSR21}.

In this work, we also wish to understand parameterized classes of strings for which we can solve trace reconstruction efficiently. Indeed, we give an algorithm using polynomial traces and runtime, that works for a general class of strings that we call $L$-separated strings. This significantly broadens the classes of strings for which polynomial-time algorithms are known~\cite{KrishnamurthyMMP21}.
%

\paragraph{Main Result}
Our main result concerns trace reconstruction of strings that are \emph{mildly separated}. We say that a string $x$ is $L$-separated if the number of zeros between any two consecutive ones is at least $L$. Depicting a string $x\in \{0,1\}^n$ with $t$ ones as 
$$
\underbrace{0 \dots 0}_{a_0 \text{ times}} 1 \underbrace{0 \dots 0}_{a_1 \text{ times}} 1 \cdots 1 \underbrace{0 \dots 0}_{a_{t} \text{ times}},
$$
it is $L$-separated if and only if $a_i\geq L$ for each $i$ with $1\leq i\leq t-1$.
Note that we make no assumptions on $a_0$ or $a_t$. Our main result is as follows.
\begin{theorem}\label{thm:main}
There exists an algorithm that solves the trace reconstruction problem with high probability in $n$ on any $L$-separated string $x$, provided that $L\geq C(\log n)^8$ for a universal constant $C$, and that the deletion probability is at most some universal constant $c_0$. The algorithm uses $N=O(n\log n)$ independently sampled traces of $x$, $\tilde{x}^{(1)}, \dots, \tilde{x}^{(N)}$, and runs in $\poly(n)$ time.
\end{theorem}
We note that the number of traces is nearly optimal. Even distinguishing between two strings $x,x'$ which contain only a single one at positions $\lfloor n/2 \rfloor$ and $\lfloor n/2 \rfloor+1$ respectively, requires $\Omega(n)$ traces to succeed with probability $1/2+\Omega(1)$.

While trace reconstruction is known to be solvable very efficiently for random strings \cite{HoldenPP18, Rubinstein23}, there are certain structured classes of strings that appear to be natural hard instances for existing approaches.  Our algorithm can be seen as solving one basic class of hard instances.  It is worth noting the work by~\cite{ChenDLSS21a} which studies the trace reconstruction problem when the deletion probability $\delta$ is sub-constant. They show that the simple Bitwise Majority Alignment (BMA) algorithm from~\cite{BatuKKM04} can succeed with $1/n^{o(1)}$ deletion probability as long as the original string does not contain \emph{deserts} \--- which are highly repetitive blocks where some short substring is repeated many times.  They then combine this with an algorithm for reconstructing repetitive blocks \--- but this part of their algorithm requires a significantly smaller deletion probability of $\delta\leq 1/n^{1/3+\eps}$.  This suggests that strings containing many repetitive blocks are a natural hard instance and good test-bed for developing new algorithms and approaches.  $L$-separated strings can be thought of as the simplest class of highly repetitive strings (where the repeating is pattern is just a $0$), where every repetition has length at least $L$.


\paragraph{Comparison to Related Work}
Most closely related to our work is the result by Krishnamurthy et al.~\cite{KrishnamurthyMMP21} stating that if $x$ has at most $k$ ones and if each pair of ones is separated by a~\emph{run} of zeros of length $\Omega(k\log n)$, then $x$ can be recovered in polynomial time from $\poly(n)$ many traces. In particular, for strings with $k=O((\log n)^7)$ ones, the required separation is milder than ours, albeit not below $\Omega(\log n)$. Our algorithm works in general assuming a $\polylog n$ separation of the ones but with no additional requirement on the number of ones: indeed, we could even have $\frac{n}{\polylog n}$ ones. Assuming no sparsity assumptions, \cite{KrishnamurthyMMP21} would need to set $L \ge \Omega(\sqrt{n \log n})$, as a $\sqrt{n \log n}$-separated string can be $\Theta(\sqrt{n/\log n})$-sparse in the worst case. The techniques of~\cite{KrishnamurthyMMP21} are also very different than ours. They recursively cluster the positions of the ones in the observed traces to correctly align a large fraction of the ones in the observed traces to ones in the string $x$. In contrast, our algorithm works quite differently and is of a more sequential nature processing the traces from left to right (or right to left). See~\Cref{sec:technical_contributions} for a discussion of our algorithm.

Another paper studying strings with large runs is by Davies et al.~\cite{DaviesRSR21}. They consider~\emph{approximate} trace reconstruction, specifically how many traces are needed to approximately reconstruct $x$ up to edit distance $\eps n$ under various assumptions on the lengths of the runs of zeros and ones in $x$. Among other results but most closely related to ours, they show that one can $\eps$-approximately reconstruct $x$ using $O((\log n)/\eps^2)$ traces if the runs of zeros have length $\gg \frac{\log n}{\eps}$ and if the runs of ones are all of length $\leq C \log n$  or $\gg 3C\log n$ for a constant $C$ (e.g. they could have length one as in our paper). However, for exact reconstruction, they would need to set $\eps < 1/n$, which means they do not provide any nontrivial guarantees in our setting.

\subsection{Technical Contributions}\label{sec:technical_contributions}
In this section, we give a high level overview of our techniques. Recall that we want to reconstruct a string $x\in \{0,1\}^n$ from independent traces $\tilde x$ where we assume that $x$ is mildly separated. More concretely, we assume that there are numbers $a_0,\dots,a_t \gg \polylog n$ such that $x$ consists of $a_0$ zeros followed by a one, followed by $a_1$ zeros followed by a one and so on, with the last $a_t$ bits of $x$ being zero. Writing $a_{\leq i}=\sum_{0\leq j\leq i} a_j$, we thus have that there are $t$ ones in $x$ at positions $a_{\leq i}+i+1$ for $0\leq i\leq t-1$. 

Note that a retained bit in a trace $\tilde x$ naturally corresponds to a bit in $x$. More formally, for a trace $\tilde x$ of length $\ell$, let $i_1<\cdots<i_\ell$ be the $\ell$ positions in $x$ where the bit was retained when generating $\tilde x$ so that $\tilde x=x_{i_1}\cdots x_{i_\ell}$. Then, the correspondence is defined by the map from $[\ell]$ to $[n]$ mapping $j\mapsto i_j$. We think of this map as the correct \emph{alignment} of $\tilde x$ to $x$.

Our main technical contribution is an alignment algorithm (see~\Cref{alg:findmthone}) which takes in some $m \leq t$ and estimates $b_0,\dots,b_{m-1}$ of $a_0,\dots,a_{m-1}$ satisfying that for all $i$, $|b_i-a_i|=O(\sqrt{a_i\log n})$, and correctly aligns the one in a trace $\tilde x$ corresponding to the $m$'th one of $x$ with probability $1-O(\delta)$ (where the randomness is over the draw of $\tilde x$ \---naturally, this requires that the $m$'th one of $x$ was not deleted). 

Moreover, we ensure that the alignment procedure has that with high probability, say $1-O(n^{-100})$, it never aligns a one in $\tilde x$ too far to the right in $x$: if the one in $\tilde x$ corresponding to the  $m_0$'th one of $x$ is aligned to the $m$'th one of $x$, then $m\leq m_0$. We will refer to this latter property by saying that the algorithm is \emph{never ahead} with high probability. If $m< m_0$, we say that the algorithm is \emph{behind}. Thus, to show that the algorithm correctly aligns the $m$'th one, it suffices to show that the probability that the algorithm is behind is $O(\delta)$.

We first discuss how to implement this alignment procedure and then afterwards we discuss how to complete the reconstruction by using this alignment procedure.

\paragraph{The alignment procedure of~\Cref{alg:findmthone}.}
The main technical challenge of this paper is the analysis of~\Cref{alg:findmthone}. Let us first describe on a high level how the algorithm works.
For $0\leq j\leq j'\leq m$, we write $b_{j:j'}=\sum_{i=j}^{j'-1}b_j$. Suppose that the trace $\tilde x$ consists of $s_0$ zeros followed by a one followed by $s_1$ zeros followed by a one and so on. The algorithm first attempts to align the first one in $\tilde x$ with a one in $x$ by finding the minimal $j_0$ such that $(1-\delta)\cdot b_{0:j_0}$ is within $C\log n\sqrt{b_{0:j_0}}$ of $s_0$ for a sufficiently large $C$. Inductively, having determined $j_i$ (that is the alignment of the $i$'th one of $\tilde x$), it looks for the minimal $j_{i+1}> j_i$ satisfying that there is a $j_i \leq j'< j_{i+1}$ such that $b_{j':j_{i+1}}\cdot(1-\delta)$ is within $C\log n\sqrt{b_{j':j_{i+1}}}$ of $s_{i+1}$.
Intuitively, when looking at the $i$'th one in the trace, we want to find the \emph{earliest} possible location in the real string (which has gaps estimated by $b_0, b_1, \dots$) that could plausibly align with the one in the trace.

It is relatively easy to check that the algorithm is never ahead with very high probability. Indeed, by concentration bounds on the number of deleted zeros and the fact that $|b_j-a_j|=O(\sqrt{a_j\log n})$ for all $j\leq m$, it always has the option of aligning the $(i+1)$'st one in $\tilde x$ to the correct one in $x$. However, it might align to an earlier one in $x$ since it is looking for the \emph{minimum} $j_{i+1}$ such that an alignment is possible. For a very simple example, suppose that $a_0=n^{\Omega(1)}$ and $a_1=\cdots=a_m=b_1=\cdots=b_m=\polylog(n)$. If the first $k<m$ ones of $x$ are deleted and the $(k+1)$'st one is retained, the algorithm will align the retained one (which corresponds to the $(k+1)$'st one of $x$) with the first one of $x$ resulting in the aligning algorithm being $k$ steps behind. Moreover, the algorithm will remain $k$ steps behind all the way up to the $m$'th one of $x$. The probability of this happening is $\Theta(\delta^k)$. To prove that the probability of the algorithm being behind when aligning the $m$'th one of $x$ is at most $1-O(\delta)$, we prove a much stronger statement which is amenable to an inductive proof, essentially stating that this is the worst that can happen: The probability of the algorithm being $k$ steps behind at any fixed point is bounded by $(C\delta)^k$ for a constant $C$.
In particular, we show that there is a sort of amortization \--- whenever there is a substring that can cause the algorithm to fall further behind with some probability (i.e. if certain bits are deleted), the substring also helps the algorithm catch back up if it is already behind.

\paragraph{Reconstructing $x$ using~\Cref{alg:findmthone}.}
Using~\Cref{alg:findmthone} we can iteratively get estimates $b_0,\dots,b_t$ with $|b_i-a_i|=O(\sqrt{a_i\log n})$. Namely, suppose that we have the estimates $b_0,\dots,b_m$. We then run~\Cref{alg:findmthone} on $O(\log n)$ independent traces and with high probability, for a $1-O(\delta)$ fraction of them, we have that the $m$'th and $(m+1)$'st one of $x$ are retained in $\tilde x$ and correctly aligned. In particular, with probability $1-O(\delta)$ we can identify both the $m$'th and $(m+1)$'st one of $x$ in $\tilde x$ and taking the median over the gaps between these (and appropriately rescaling by $\frac{1}{1-\delta}$), we obtain an estimate of $b_{m+1}$ such that $|b_{m+1}-a_{m+1}|=O(\sqrt{a_{m+1}\log n})$). Note that the success probability of $1-O(\delta)$ is enough to obtain the coarse estimates using the median approach but we cannot obtain a fine estimate by taking the average since with constant probability $O(\delta)$, we may have misaligned the gap completely and then our estimate can be arbitrarily off.

To obtain fine estimates, we first obtain coarse estimates, say $b_0,\dots, b_t$, for all of the gaps. Next, we show that we can identify the $m$'th and $(m+1)$'st one in $x$ in a trace $\tilde x$ (if they are retained) and we can detect if they were deleted not just with probability $1-O(\delta)$ but with very high probability. The trick here is to run~\Cref{alg:findmthone} both from the left and from the right on $\tilde x$ looking for respectively the one in $\tilde x$ aligned to the $m$'th one in $x$ and the one in $\tilde x$ aligned to the $(m+1)$'st one in $x$ (which is the $(t-m)$'th one when running the algorithm from the right). If either of these runs fails to align a one in $\tilde x$ to respectively the $m$'th and $(m+1)$'st one in $x$ or the runs disagree on their alignment,
then we will almost certainly know. To see why, assuming that we are never ahead in the alignment procedure from the left, if we believe we have reached the $m$'th one in $x$, then we are truly at some $m_0$'th one where $m_0 \ge m$. By a symmetric argument, if we believe we have reached the $(m+1)$'st one in $x$ after running the procedure from the right, we are truly at the $m_1$'th one in $x$, where $m_1 \le m+1$. The key observation now is that $m_0 \le m_1$ \emph{if and only if} $m_0 = m$ and $m_1 = m+1$, meaning that both runs succeeding is equivalent to the one found in the left-alignment procedure being strictly earlier than the one found in the right-alignment procedure.
So, if we realize that either run fails to align the ones properly, we discard the trace and repeat on a newly sampled trace.


%
%
Finally, we can ensure that the success of the runs of the alignment algorithm is independent of the deletion of zeros between the $m$'th and $(m+1)$'st ones in $x$.  If a trace is not discarded, then with very high probability, the gap between the ones in $\tilde x$ aligned to the $m$'th and $(m+1)$'st ones in $x$ (normalized by $\frac{1}{1-\delta}$) is an unbiased estimator for $a_{m+1}$. By taking the average of the gap over $\tilde O(n)$ traces, normalizing by $\frac{1}{1-\delta}$, and rounding to the nearest integer, we determine $a_{m+1}$ exactly with very high probability. Doing so for each $m$, reconstructs $x$.


\paragraph{Road map of our paper}
In~\Cref{sec:notation}, we introduce notation. In~\Cref{sec:main-process}, we describe and analyse our main alignment procedure. We first prove that with high probability it is never ahead (\Cref{lem:not-ahead}). Second, in \Cref{sec:main-technical}, we bound the probability that it is behind (\Cref{lem:main-technical}). Finally, in~\Cref{sec:algorithm}, we describe our full trace reconstruction algorithm and prove~\Cref{thm:main}.

\section{Notation}\label{sec:notation}

We note a few notational conventions and definitions.

\begin{itemize}
    \item We recall that a bitstring $x$ is \emph{$L$-separated} if the gap between any consecutive $1$'s in the string contains at least $L$ $0$'s.
    \item Given an string $x$, we say that a \emph{run} is a contiguous sequence of $0$'s in $x$. For $x = \underbrace{0 \dots 0}_{a_0 \text{ times}} 1 \underbrace{0 \dots 0}_{a_1 \text{ times}} 1 \cdots 1 \underbrace{0 \dots 0}_{a_{t} \text{ times}},$ the $i$th run of $x$ is the sequence $\underbrace{0 \dots 0}_{a_i \text{ times}}$, and has length $a_i$.
    \item For any bitstring $x = x_1 x_2 \cdots x_n$, we use $\rev(x) := x_n x_{n-1} \cdots x_1$ to denote the string where the bits have been reversed.
    \item We use $\ba = a_0, a_1, \dots, a_{m-1}$ to denote an integer sequence of length $m$. For notational convenience, for any $0 \le j < j' \le m$, we write $\ba_{j:j'}$ to denote the subsequence $a_{j}, a_{j+1}, \dots, a_{j'-1}$, and $a_{j:j'} := \sum_{i = j}^{j'-1} a_i$.
\end{itemize}

We will define some sufficiently large constants $C_0, C_1, C_2, C_3$ and a small constant $c_0$. We will assume the separation parameter $L = C_3 \cdot \log^8 n$, and the deletion parameter $\delta \le c_0$, where $c_0 = \frac{1}{3 \cdot 10^6}$. We did not make significant effort to optimize the constant $c_0$ or the value $8$ in $\log^8 n$, though we believe that any straightforward modifications to our analysis will not obtain bounds such as $c_0 \ge \frac{1}{2}$ or a separation of $L = O(\log n)$.

\section{Main Alignment Procedure} \label{sec:main-process}

\subsection{Description and Main Lemmas} \label{sec:main-process-description}

In this section, we consider a probabilistic process that models a simpler version of the trace reconstruction problem that we aim to solve. In the simpler version of the trace reconstruction problem, suppose that we never delete any $0$'s, but delete each $1$ independently with $\delta$ probability. Let $a_0, \dots, a_{m-1} \le n$ represent the true lengths of the first $m$ gaps (so the first $1$ is at position $1+a_0$, the second $1$ is at position $2 + a_0 + a_1$, and so on). Moreover, suppose we have some current predictions $b_0, \dots, b_{m-1} \le n$ of the gaps $a_0, \dots, a_{m-1}$. The high level goal will be, given a single trace (where the trace means only $1$s are deleted), to identify the $m$th $1$ in the trace from the the original string with reasonably high probability. (Note that the $m$th $1$ is deleted with $\delta$ probability, in which case we cannot succeed.)

We will describe and analyze the probabilistic process, and then explain how the analysis of the process can help us solve the trace reconstruction problem in \Cref{sec:algorithm}.

In the process, we fix $m \le n$ and two sequences $\ba = a_0, \dots, a_{m-1}$ and $\bb = b_0, \dots, b_{m'-1}$ where $\ba$ has length $m$ but $\bb$ has some length $m'$ which may or may not equal $m$.
Moreover, we assume $L \le a_i \le n$ and $L \le b_j \le n$ for every term $a_i \in \ba$ and $b_j \in \bb.$ 

Now, for each $1 \le i \le m-1$, let $w_i \in \{0, 1\}$ be i.i.d. random variables, with $w_i = 1$ with $1-\delta$ probability and $w_i = 0$ with $\delta$ probability. Also, let $w_0 = w_m = 1$ with probability $1$.
For each $0 \le i \le m$ with $w_i = 1$, we define a value $f_i$ as follows. First, we set $f_0 = 0$. Next, for each index $i \ge 1$ such that $w_i = 1$, let $i_0$ denote the previous index with $w_{i_0} = 1$. We define $f_i$ to be the smallest index $j' > f_{i_0}$ such that there exists $f_{i_0} \le j < j'$ with $|b_{j:j'} - a_{i_0:i}| \le C_0 \cdot \log n \cdot \sqrt{b_{j:j'}}$, where $C_0$ is a sufficiently large constant. (If such an index does not exist, we set $f_i = \infty$.) 

Our goal will be for $f_m = m$. In general, for any $i$ with $w_i = 1$, we would like $f_i = i$. If $f_i < i$, we say that we are $i-f_i$ steps behind at step $i$, and if $f_i > i$, we say that we are $f_i-i$ steps ahead at step $i$.

First, we note the following lemma, which states that we will never be ahead with very high probability, as long as the sequences $\ba$ and $\bb$ are similar enough.

\begin{lemma} \label{lem:not-ahead}
    Set $C_1 = C_0/4$. Let $\ba, \bb$ be sequences of lengths $m, m'$, respectively, where $m' \ge m$.
    Suppose that $|b_{i}-a_{i}| \le C_1 \cdot \sqrt{b_{i} \log n}$ for all $0 \le i < m$. Then, with probability at least $1 - \frac{1}{n^{10}}$ (over the randomness of the $w_i$), for all $0 \le i \le m$ with $w_i = 1$, $f_i \le i$.
\end{lemma}

\begin{proof}
    Let us consider the event that for every index $0 \le i \le m - 15 \log n$, at least one of $w_i, w_{i+1}, \dots, w_{i+15 \log n}$ equals $1$. Equivalently, the string $w_0 w_1 \cdots w_m$ does not ever have $15 \log n + 1$ $0$'s in a row. For any fixed $i$, the probability of this being false is at most $\delta^{15 \log n} \le n^{-15}$, so by a union bound over all choices of $i$, the event holds with at most $n^{-10}$ failure probability.

    First, note that $f_0 = 0$. Now, suppose that some $i \ge 0$ satisfies $w_i = 1$ and $f_i \le i$. Suppose $i'$ is the smallest index strictly larger than $i$ such that $w_{i'} = 1$. Note that $i'-i \le 15 \log n + 1 \le 16 \log n$, by our assumed event. Note that if we set $j = i$ and $j' = i'$, then $j' > j \ge f_i$, since $f_i \le i$.
    Moreover, $|b_{j:j'}-a_{j:j'}| \le \sum_{i=j}^{j'-1} |b_i-a_i|\le C_1 \sqrt{\log n} \cdot \sum_{i=j}^{j'-1} \sqrt{b_i} \le C_1 \cdot \sqrt{\log n} \cdot \sqrt{b_{j:j'} \cdot |j'-j|} \le 4 C_1 \cdot \log n \cdot \sqrt{b_{j:j'}}$, where the second to last inequality is by Cauchy-Schwarz. Thus, $j = i, j' = i'$ satisfies the requirements for $f_{i'}$, which means that $f_{i'} \le j' = i'$. Thus, if $f_i \le i$, $f_{i'} \le i'$. Since $f_0 \le 0$, this means $f_i \le i$ for all $i$ with $w_i = 1$.
\end{proof}

The main technical result will be showing that $f_m \ge m$ with reasonably high probability, i.e., with reasonably high probability we are not behind. This result will hold for \emph{any} choice of $\ba, \bb$ and does not require any similarity between these sequences. In other words, our goal is to prove the following lemma.

\begin{lemma} \label{lem:main-technical}
    Let $\ba, \bb$ be strings of length at most $n$ with every $\ba_i, \bb_j$ between $L$ and $n$, where $L = C \cdot \log^8 n$ for a sufficiently large constant $C$. Define $m = |\ba|$.
    Then, for any $\delta \le \frac{1}{3 \cdot 10^6}$, with probability at least $1-200 \cdot \delta$ over the randomness of $w_1, \dots, w_{m-1}$, $f_m \ge m$.
\end{lemma}

\subsection{Proof of \Cref{lem:main-technical}}\label{sec:main-technical}

In this section, we prove \Cref{lem:main-technical}.

We will set a parameter $K = C_2 \log n$, where $C_2$ is a sufficiently large constant.
For any $k \ge 0$, given the sequences $\ba = a_0, \dots, a_{m-1}$ and $\bb = b_0, \dots, b_{m'-1}$ (of possibly differing lengths), we define $p_k(\ba, \bb)$ to be the probability (over the randomness of $w_1, \dots, w_{m-1}$) that
\begin{itemize}
    \item $f_m \le m-k$.
    \item For any indices $0 \le i \le i' \le m$ with $w_i, w_{i'} = 1$, $f_{i'}-f_i \ge (i'-i)-K$.
\end{itemize}
Equivalently, this is the same as the probability that we fall behind at least $k$ steps from step $0$ to step $m$, but we never fall behind $K+1$ or more steps (relatively) from any (possibly intermediate) steps $i$ to $i'$.
For any $m \ge 1$, we define $p_k(m)$ to be the supremum value of $p_k(\ba, \bb)$ over any sequences $\ba, \bb$ where $\ba$ has length at most $m$ and every $a_i$ and $b_j$ is between $L$ and $n$, and we also define $p_k := \sup_{m \ge 1} p_k(m)$.

Note that for any $k > K$, $p_k(\ba, \bb) = 0$, as $f_m 
= m-k$ means $f_m - f_0 < (m-0) - K$. So, $p_k(m)$ and $p_k$ also equal $0$ for any $k > K$.

First, we note a simple proposition, that will only be useful for simplifying the argument at certain places.

\begin{proposition} \label{prop:p0-equals-1}
    For any $m \ge 1$, $p_0(m) = 1$.
\end{proposition}

\begin{proof}
    Since $p_0(m)$ is the maximum over all $\ba, \bb$ where $\ba$ has length at most $m$, it suffices to prove it for some $\ba, \bb$ of length $1$. Indeed, for $m = 1$ and $a_0 = b_0 = L$, we must have that $w_0 = w_1 = 1$, so we must have $f_0 = 0$ and $f_1 = 1$.
\end{proof}

We now aim to bound the probabilities $p_k$ for $k \le K$. We will do this via an inductive approach on the length of $m$, where the high-level idea is that if we fall back by $k$ steps, there is a natural splitting point where we can say first we fell back by $k_1$ steps, and then by $k_2$ steps, for some $k_1, k_2 > 0$ with $k_1+k_2 = k$ -- see Lemmas~\ref{lem:half-periodic-bound} and~\ref{lem:no-periodic-bound}. This natural splitting point will be based on the structure of the similarity of $\ba$ and $\bb$, and will not work if $\ba$ and $\bb$ share a $k$-periodic structure. But in the periodic case, we can give a more direct argument that we cannot fall back by $k$ steps (i.e., a full period), even with $\frac{1}{\poly(n)}$ probability -- see \Cref{lem:periodic-bound}. We can then compute a recursive formula for the probability of falling back $k$ steps, by saying we need to first fall back $k_1$ steps and then fall back $k_2$ steps. In \Cref{lem:computation}, we bound the terms of this recursion.

\begin{lemma} \label{lem:periodic-bound}
    Fix any $m \ge k \ge 1$ such that $k \le K$, and suppose that $L \ge C_3 \cdot \log^8 n$, where $C_3$ is a sufficiently large multiple of $C_0^2 \cdot C_2^6$. Suppose that $\ba, \bb$ are sequences such that for every $0 \le i < m-k,$ $|b_i-a_i| \le C_0 \log n \cdot \sqrt{b_i}$ and $|b_i-a_{i+k}| \le C_0 \log n \cdot \sqrt{b_i}$. Then, the probability $p_k(\ba, \bb) \le (2\delta)^{K}$.
\end{lemma}

\begin{proof}
    We show that the probability of ever being behind by $k$ or more is at most $(2\delta)^{K}$. In fact, we will show this deterministically never happens, conditioned on the event that for every index $0 \le i \le m- K \cdot k$, at least one of $w_i, w_{i+k}, w_{i+2k}, \dots, w_{i+ K \cdot k}$ equals $1$. Indeed, the probability of this being false for any fixed $i$ is at most $\delta^{K}$, so by a union bound over all choices of $i$, the event holds with at most $n \cdot \delta^K \le (2\delta)^{K}$ failure probability.

    Now, assume the event and suppose that $f_i \le i-k$ holds for some $i$. More precisely, we fix $i$ to be the smallest index such that $w_i = 1$ and $f_i \le i-k$.
    
    First, assume that $i \ge 2K \cdot k$. 
    Consider the values $a_{i-1}, a_{i-2}, \dots, a_{i-k}$, and let $h = \argmax_{1 \le t \le k} a_{i-t}.$ By our conditional assumption, and since $i \ge 2K \cdot k$, at least one of $w_{i-h}, w_{i-h-k}, \dots, w_{i-h-K \cdot k}$ equals $1$. Say that $w_{i-h-r \cdot k} = 1$, where $0 \le r \le K$. Also, by our choice of $i$, we know that $f_{i-h-r \cdot k} > i-h-(r+1) \cdot k \ge 0$, and that $f_i \le i-k$. So, we have two options:
\begin{enumerate}
    \item $i \ge 2K \cdot k$, and $f_i \le i-k ,\, f_{i-h-r \cdot k} > i-h-(r+1) \cdot k \ge 0$, for some $r \le K$ and where $h = \argmax_{1 \le t \le k} a_{i-t}$.
    \item $i < 2K \cdot k$, and $f_i \le i-k,\, f_0 = 0$.
\end{enumerate}

    Now, let's consider the list of all indices $i_0 < i_1 < \cdots < i_s = i$ with $w_{i_0}, w_{i_1}, \dots, w_{i_s} = 1$, starting with $i_0 = i-h-r \cdot k$ if $i \ge 2 K \cdot k$ and $i_0 = 0$ otherwise, and ending with $i_s = i$. By definition of the sequence $f$, for every $0 \le t < s$ there exists $j, j'$ such that $f_{i_t} \le j < j' \le f_{i_{t+1}}$ and $|b_{j:j'} - a_{i_t:i_{t+1}}| \le C_0 \log n \cdot \sqrt{b_{j:j'}}$. Assuming that $L \ge (10 C_0 \log n)^2$, then $a_{i_t:i_{t+1}} \ge (10 C_0 \log n)^2,$ which means $a_{i_t:i_{t+1}} \ge b_{j:j'}/2,$ and thus $|b_{j:j'} - a_{i_t:i_{t+1}}| \le 2 C_0 \log n \cdot \sqrt{a_{i_t:i_{t+1}}}.$ So, 
\[b_{f_{i_t}:f_{i_{t+1}}} \ge b_{j:j'} \ge a_{i_t:i_{t+1}} - 2C_0 \log n \sqrt{a_{i_t}:a_{i_{t+1}}}.\]
    Adding the above equation over $0 \le t \le s-1$, we obtain 
\[b_{f_{i_0}:f_i} \ge a_{i_0:i} - 2 C_0 \log n \cdot \sum_{t=0}^{s-1} \sqrt{a_{i_t:i_{t+1}}} \ge a_{i_0:i} - 2 C_0 \log n \cdot \sqrt{a_{i_0:i} \cdot s},\]
    where the final line follows by Cauchy-Schwarz. Let $j_0$ be $i-h-(r+1) \cdot k+1$ if $i \ge 2K \cdot k$ and $j_0 = 0$ otherwise. Then, since $s \le i_s-i_0 \le 2 k \cdot K \le 4K^2$, we have
\begin{equation} \label{eq:b-bigger-than-a}
    b_{j_0:i-k} \ge b_{f_{i_0}:f_i} \ge a_{i_0:i} - 4 C_0 \cdot K \log n \cdot \sqrt{a_{i_0:i}}.
\end{equation}

    The above equation tells us that $b_{j_0:i-k} = \sum_{t = j_0}^{i-k-1} b_t$ can't be too much smaller than $a_{i_0:i} = \sum_{t = i_0}^{i-1} a_t$. We now show contrary evidence, thus establishing a contradiction.

    First, we compare $b_{j_0:i-k}$ to $a_{j_0+k:i}$. Indeed, for any $t < i \le m$, $|b_{t-k}-a_t| \le C_0 \log n \cdot \sqrt{b_{t-k}}$. Since every $a_i \ge (10 C_0 \log n)^2$, this also means $|b_{t-k}-a_t| \le 2 C_0 \log n \cdot \sqrt{a_t}$. Adding over all $j_0 \le t < i-k$, we have 
\[a_{j_0+k:i} \ge b_{j_0:i-k} - 2 C_0 \log n \cdot \sum_{t=j_0+k}^{i-1} \sqrt{a_t} \ge b_{j_0:i-k} - 4 C_0 \cdot K \log n \cdot \sqrt{a_{j_0+k:i}},\]
    where the last inequality follows by Cauchy-Schwarz and the fact that $i-(j_0+k) \le i-j_0 \le 2K \cdot k \le 4K^2$.

    However, we do not care about $a_{j_0+k:i}$ -- we really care about $a_{i_0:i}$. To bound this, first note that for any $k \le i < m$, $|a_i-b_{i-k}| \le C_0 \log n \cdot \sqrt{b_i}$ and $|a_{i-k}-a_{i-k}| \le C_0 \log n \cdot \sqrt{b_i}$. So, $|a_i-a_{i-k}| \le 4 C_0 \log n \cdot \sqrt{a_i}$, assuming every $a_i \ge (10 C_0 \log n)^2$.
    If we additionally have that $L \ge (100 C_0 \log n \cdot K)^2,$ then $|a_i-a_{i-s \cdot k}| \le 8 C_0 \log n \cdot s \cdot \sqrt{a_i}$ for any $1 \le s \le K$ and $s \cdot k \le i < m$. Importantly, $\frac{a_{i-s \cdot k}}{a_i} \in [1/2, 2]$.
    
    In the case that $i \ge 2K \cdot k,$ this implies that $\sum_{t = i-h-r \cdot k}^{i-1} a_t \le 2 (r+1) \cdot \sum_{t = i-k}^{i-1} a_t \le 4K \cdot \sum_{t = i-k}^{i-1} a_t$. So, because $h = \argmax_{1 \le t \le k} a_{i-t},$ we have
    \[a_{i_0} = a_{i-h-r \cdot k} \ge \frac{1}{2} \cdot a_{i-h} \ge \frac{1}{2k} \cdot \sum_{t=1}^k a_{i-t} \ge \frac{1}{8K^2} \cdot \sum_{t = i-h-r \cdot k}^{i-1} a_t.\]
    Recalling that $i_0 = i-h-r \cdot k$ and $j_0 = i-h-(r+1) \cdot k + 1$, since $i_0 = j_0+k-1$, 
\begin{align} \label{eq:a-bigger-than-b}
    a_{i_0:i} \ge \left(1+\frac{1}{8K^2}\right) \cdot a_{j_0+k:i} &\ge \left(1+\frac{1}{8K^2}\right) \cdot (b_{j_0:i-k} - 4C_0 \cdot K \log n \cdot \sqrt{a_{j_0+k:i}}) \nonumber \\
    &\ge \left(1+\frac{1}{8K^2}\right) \cdot (b_{j_0:i-k} - 4C_0 \cdot K \log n \cdot \sqrt{a_{i_0:i}}).
\end{align}
    
    In the case that $i < 2K \cdot k$, we instead have $\sum_{t=0}^{i-1} a_t \le 2 \cdot \left\lceil \frac{i}{k}\right\rceil \cdot \sum_{t=0}^{k-1} a_t \le 4K \cdot \sum_{t=0}^{k-1} a_t$. So, since $i_0 = j_0 = 0$, we have that 
\[a_{i_0:i} = a_{j_0+k:i}+a_{0:k} \ge \left(1 + \frac{1}{4K}\right) \cdot a_{j_0+k:i} \ge \left(1+\frac{1}{4K}\right) \cdot (b_{j_0:i-k} - 4C_0 \cdot K \log n \cdot \sqrt{a_{j_0+k:i}}),\]
    so the same bound as \eqref{eq:a-bigger-than-b} holds (in fact, an even stronger bound holds).

    So, both \eqref{eq:b-bigger-than-a} and \eqref{eq:a-bigger-than-b} hold, in either case. Together, they imply that
\[a_{i_0:i} \ge \left(1+\frac{1}{8K^2}\right) \cdot \left(a_{i_0:i} - 8C_0 \cdot K \log n \cdot \sqrt{a_{i_0:i}}\right).\]
    This is impossible if $a_{i_0:i}$ is a sufficiently large multiple of $(C_0 \cdot K \log n \cdot K^2)^2 = C_0^2 \cdot \log^2 n \cdot K^6$. Since $i \ge i_0+1$ in either case, it suffices for $L$ to be a sufficiently large multiple of $C_0^2 \cdot \log^2 n \cdot K^6 = C_0^2 C_2^6 \cdot \log^8 n$.
\end{proof}

\begin{lemma} \label{lem:half-periodic-bound}
    Fix any $m \ge k$ such that $k \le K$, and suppose that $L \ge C_3 \cdot \log^2 n \cdot K^6$. Suppose that $\ba, \bb$ are sequences of length $m$, such that for every $0 \le i < m-k,$ $|b_i-a_{i+k}| \le C_0 \log n \cdot \sqrt{b_i}$.
    Then, the probability
\[p_k(\ba, \bb) \le (2\delta)^{K} + \sum_{\substack{h_1, h_2, k_2, k_2 \ge 0 \\ h_1+h_2+k_1+k_2 \ge k \\ k_1, k_2 \le K \\ (h_1, h_2, k_1, k_2) \neq (0, 0, 0, k), (0, 0, k, 0)}} \delta^{h_1+h_2} p_{k_1}(m-1) p_{k_2}(m-1).\]
\end{lemma}

\begin{proof}
    Suppose that for all $0 \le i < m-k,$ $|b_i-a_i| \le C_0 \log n \sqrt{b_i}$. Then, we can use Lemma~\ref{lem:periodic-bound} to bound $p_k(\ba, \bb) \le (2\delta)^{K}$. Alternatively, let $0 \le h < m-k$ be the smallest index such that $|b_h-a_h| > C_0 \log n \cdot \sqrt{b_h}$. Next, let $h_1, h_2 \ge 0$ be such that $h-h_1$ is the largest index less than $h$ with $w_{h-h_1} = 1$, and $h+1+h_2$ is the smallest index at least $h+1$ with $w_{h+1+h_2} = 1$. Finally, let $k_1 := \max(0, (h-h_1) - f_{h-h_1})$ and $k_2 := \max(0, (m-(h+1+h_2)) - (f_m-f_{h+1+h_2}))$. In other words, $k_1$ is the number of steps we fall behind from $0$ to $h-h_1$, and $k_2$ is the number of steps we fall behind from $h+1+h_2$ to $m$.

    Note that $k_1+k_2 \ge m-1-h_1-h_2 - f_m + f_{h+1+h_2}-f_{h-h_1}$, and since each subsequent $f_i$ is strictly increasing, this means $f_{h+1+h_2}-f_{h-h_1} \ge 1$, so $k_1+k_2 \ge m - f_m - (h_1+h_2) \ge k - (h_1+h_2)$, assuming that $f_m \le m-k$. In other words, we have that $h_1, h_2, k_1, k_2$ are nonnegative integers such that $h_1+h_2+k_1+k_2 \ge k$.

    Now, let us bound the probability (over the randomness of $w_1, \dots, w_{m-1}$) of the event indicated by $p_k(\ba, \bb)$ occurring, with the corresponding values $h_1, h_2, k_1, k_2$. Note that for any fixed $h_1, h_2$, the event of those specific values is equivalent to $w_{h-h_1}$ and $w_{h+1+h_2}$ being $1$, and everything in between being $0$. So, the probability is at most $\delta^{h_1+h_2}$. Now, conditioned on $h_1, h_2$, the values $k_1, k_2$ imply that we fall back $k_1$ steps from step $0$ to $h-h_1$ (or we may move forward if $k_1 = 0$) and we fall back $k_2$ steps from step $h+1+h_2$ to $m$. Moreover, there cannot be two steps $i, i'$ such that that we fell back $K$ steps from $i$ to $i'$. Since $h-h_1 \le h < m$ and $m-(h+1+h_2) \le m-1$, this means both $h-h_1, m-(h+1+h_2) \le m-1$. So, the overall probability of the corresponding values $h_1, h_2, k_1, k_2$ is at most $\delta^{h_1+h_2} \cdot p_{k_1}(m-1) \cdot p_{k_2}(m-1)$, where we are using the fact that $p_0(m) = 1$ for all $m$ by \Cref{prop:p0-equals-1}.

    Overall, the probability $p_k(\ba, \bb)$ is at most 
\[\sum_{\substack{h_1, h_2, k_1, k_2 \ge 0 \\ h_1+h_2+k_1+k_2 \ge k \\ k_1, k_2 \le K}} \delta^{h_1+h_2} \cdot p_{k_1}(m-1) p_{k_2}(m-1).\]
    We can cap $k_1, k_2$ as at most $K$ since otherwise $p_{k_1}(m-1)$ or $p_{k_2}(m-1)$ is $0$. Moreover, we can give improved bounds in the cases when $h_1 = h_2 = 0$ and either $(k_1, k_2) = (0, k)$ or $(k_1, k_2) = (k, 0)$.

    Note that in either case, both $w_h$ and $w_{h+1}$ equal $1$. In the former case, we must have $f_{h} = h-k$ and $f_{h+1} = h+1-k$. Importantly, the algorithm fell back by exactly $k$ steps from $0$ to $h$, However, we know that for all $0 \le i \le h-1$, $|b_i-a_i| \le C_0 \log n \cdot \sqrt{b_i}$. In that case, if we restrict ourselves to the strings $\ba_{0:h} = a_0 a_1 \cdots a_{h-1}$ and $\bb_{0:h} = b_0 b_1 \cdots b_{h-1}$, we are dealing with the case of Lemma~\ref{lem:periodic-bound}. Hence, we can bound the overall probability of this case by $(2\delta)^{K}$. In the latter case, we must have $f_{h} = h$ and $f_{h+1} = h+1$, since we need to fall back by exactly $k$ steps from $h$ to $m$. However, this actually cannot happen, because by definition of $f_h$ and $f_{h-1}$, we must have that $|b_h-a_h| \le C_0 \log n \cdot \sqrt{b_h},$ which is not true by our definition of $h$.

    Overall, this means
\[p_k(\ba, \bb) \le (2\delta)^{K} + \sum_{\substack{h_1, h_2, k_2, k_2 \ge 0 \\ h_1+h_2+k_1+k_2 \ge k \\ k_1, k_2 \le K \\ (h_1, h_2, k_1, k_2) \neq (0, 0, 0, k), (0, 0, k, 0)}} \delta^{h_1+h_2} p_{k_1}(m-1) p_{k_2}(m-1). \qedhere\]
\end{proof}

\begin{lemma} \label{lem:no-periodic-bound}
    Fix any $m \ge k$ such that $k \le K$, and suppose that $L \ge C_3 \cdot \log^2 n \cdot K^6$. Suppose that $\ba, \bb$ are sequences of length $m$.
    Then, the probability
\[p_k(\ba, \bb) \le (2\delta)^{K} + \sum_{\substack{h_1, h_2, k_2, k_2 \ge 0 \\ h_1+h_2+k_1+k_2 \ge k \\ k_1, k_2 \le K \\ (h_1, h_2, k_1, k_2) \neq (0, 0, 0, k), (0, 0, k, 0)}} \delta^{h_1+h_2} p_{k_1}(m-1) p_{k_2}(m-1).\]
\end{lemma}

\begin{proof}
    Our proof will be quite similar to that of \Cref{lem:half-periodic-bound}, so we omit some of the identical details.
    
    First, assume that for every $k \le i < m,$ $|b_{i-k}-a_{i}| \le C_0 \log n \cdot \sqrt{b_{i-k}}$. Then, we can directly apply \Cref{lem:half-periodic-bound}. Alternatively, let $k \le h < m$ be the largest index such that $|b_{h-k}-a_{h}| > C_0 \log n \cdot \sqrt{b_{h-k}}$. 
    As in the proof of \Cref{lem:half-periodic-bound}, let $h_1, h_2 \ge 0$ be such that $h-h_1$ is the largest index less than $h$ with $w_{h-h_1} = 1$, and $h+1+h_2$ is the smallest index at least $h+1$ with $w_{h+1+h_2} = 1$. Also, let $k_1 := \max(0, (h-h_1) - f_{h-h_1})$ and $k_2 := \max(0, (m-(h+1+h_2)) - (f_m-f_{h+1+h_2}))$.
    
    As in the proof of \Cref{lem:half-periodic-bound}, we have $h_1+h_2+k_1+k_2 \ge k$, as long as $f_m \le m-k$. We can again do the same casework on $h_1, h_2, k_1, k_2$, to obtain
\[\sum_{\substack{h_1, h_2, k_1, k_2 \ge 0 \\ h_1+h_2+k_1+k_2 \ge k \\ k_1, k_2 \le K}} \delta^{h_1+h_2} \cdot p_{k_1}(m-1) p_{k_2}(m-1).\]

    Again, we wish to consider the individual cases of $(h_1, h_2, k_1, k_2) = (0, 0, 0, k)$ or $(h_1, h_2, k_1, k_2) = (0, 0, k, 0)$ separately. In either case, $w_h = w_{h+1} = 1$. In the former case, must have $f_h = h$ and $f_{h+1} = h+1$. In this case, from step $h+1$ to $m$ we fall behind $k$ steps. In other words, we can restrict ourselves to the strings $\ba_{h+1:m} = a_{h+1} \cdots a_{m-1}$ and $\bb_{h+1:m} = b_{h+1} \cdots b_{m-1}$. However, we have now restricted ourselves to strings which satisfy the conditions of \Cref{lem:half-periodic-bound}, so we can bound the probability in this case as at most
\[(2\delta)^{K} + \sum_{\substack{h_1, h_2, k_2, k_2 \ge 0 \\ h_1+h_2+k_1+k_2 \ge k \\ k_1, k_2 \le K \\ (h_1, h_2, k_1, k_2) \neq (0, 0, 0, k), (0, 0, k, 0)}} \delta^{h_1+h_2} p_{k_1}(m-1) p_{k_2}(m-1).\]
    In the latter case, we must have $f_h = h-k$ and $f_{h+1} = h+1-k$. However, this is impossible, because $|a_h-b_{h-k}| > C_0 \log n \cdot \sqrt{b_{h-k}},$ by our definition of $h$.

    Overall, by adding all cases together, we obtain
\[p_k(\ba, \bb) \le (2\delta)^{K} + 2 \cdot \sum_{\substack{h_1, h_2, k_2, k_2 \ge 0 \\ h_1+h_2+k_1+k_2 \ge k \\ k_1, k_2 \le K \\ (h_1, h_2, k_1, k_2) \neq (0, 0, 0, k), (0, 0, k, 0)}} \delta^{h_1+h_2} p_{k_1}(m-1) p_{k_2}(m-1). \qedhere\]
\end{proof}

Overall, this implies that 
\[p_k(m) \le (2\delta)^{K} + 2 \cdot \sum_{\substack{h_1, h_2, k_2, k_2 \ge 0 \\ h_1+h_2+k_1+k_2 \ge k \\ k_1, k_2 \le K \\ (h_1, h_2, k_1, k_2) \neq (0, 0, 0, k), (0, 0, k, 0)}} \delta^{h_1+h_2} p_{k_1}(m-1) p_{k_2}(m-1).\]

We now can universally bound $p_k$ for all $0 \le k \le K$. To do so, we first recall some basic properties of the Catalan numbers.

\begin{fact} \label{fact:catalan}
    For $n \ge 0$, the Catalan numbers $\fC_n$~\footnote{We use $\fC_n$ rather than the more standard $C_n$ to avoid confusion with the constants $C_0, C_1, \dots$ we have defined.} are defined as $\fC_n = {2n \choose n}/(n+1)$. They satisfy the following list of properties.
\begin{enumerate}
    \item $\fC_0 = 1$ and for all $n \ge 0$, $\fC_{n+1} = \sum_{i=0}^n \fC_i \fC_{n-i}$.
    \item For all $n \ge 1$, $2 \le \frac{\fC_{n+1}}{\fC_n} \le 4$.
    \item For all $n \ge 0$, $\fC_n \le 4^n$.
\end{enumerate}
\end{fact}

\begin{lemma} \label{lem:computation}
    Assume $\delta \le \frac{1}{3 \cdot \ten^3}$, and define $\fD_k := 100^{2k-1} \fC_k$ for $k \ge 1$ and $\fD_0 = 1$. Then, for all $0 \le k \le K$, $p_k \le \fD_k \cdot \delta^k$.
\end{lemma}

\begin{proof}
    We prove the statement by induction on $m$.
    For $m = 1$, note that $w_0 = w_1 = 1$ with probability $1$, so $f_1 \ge 1$. Indeed, either $f_1 = 1$ or $f_1 = \infty$. So, $p_0(m) \le 1$ and $p_k(m) = 0$ for all $k \ge 1$.

    Now, suppose that the induction hypothesis holds for $m-1$: we now prove the statement for $m$. First, note that $p_0(m) = 1 = \fD_0 \cdot \delta^0$. Next, for $k \ge 1$,
\begin{align} \label{eq:pk-bash}
    p_k(m) 
    &\le (2\delta)^{K} + 2 \cdot \sum_{\substack{h_1, h_2, k_2, k_2 \ge 0 \\ h_1+h_2+k_1+k_2 \ge k \\ k_1, k_2 \le K \\ (h_1, h_2, k_1, k_2) \neq (0, 0, 0, k), (0, 0, k, 0)}} \delta^{h_1+h_2} \cdot \fD_{k_1} \delta^{k_1} \cdot \fD_{k_2} \delta^{k_2} \nonumber \\
    &\le (2 \delta)^k + 2 \cdot \sum_{\substack{h_1, h_2, k_2, k_2 \ge 0 \\ h_1+h_2+k_1+k_2 \ge k \\ (h_1, h_2, k_1, k_2) \neq (0, 0, 0, k), (0, 0, k, 0)}} \fD_{k_1} \fD_{k_2} \cdot \delta^{h_1+h_2+k_1+k_2}.
\end{align}

    We now bound the summation in the above expression. First, we focus on the terms where one of $k_1$ or $k_2$ is $0$. If $k_1 = k_2 = 0$, the summation becomes $\sum_{h_1+h_2 \ge k} \delta^{h_1+h_2}$. If we fix $h_3 = h_1+h_2$, for each $h_3 \ge k$ there are $h_3+1$ choices of $h_1+h_2$, which means the summation is $\sum_{h_3 \ge k} \delta^{h_3} (h_3+1)$. For $\delta \le \frac{1}{3 \cdot \ten^3}$, each term is at most half the previous term, so this is at most $2(k+1) \cdot \delta^k$. Next, for $k_1 = 0, k_2 > 0$, if we fix $h_3 = h_1+h_2$, the summation is $\sum_{h_3+k_2 \ge k, (h_3, k_2) \neq (0, k)} (h_3+1) \fD_{k_2} \delta^{h_3+k_2}$, since there are $h_3+1$ choices of $(h_1, h_2): h_1+h_2 = h_3$. We have a symmetric summation for $k_1 > 0, k_2 = 0$. Finally, if we focus on the terms with $k_1, k_2 \ge 1$, by writing $h_3 = h_1+h_2$ and $k_3 = k_1+k_2$, for any fixed $h_3, k_3$, the sum of $\fD_{k_1} \fD_{k_2}$ is at most $\ten^{2k_1+2k_2-2} \cdot \fC_{k_3+1} \le \ten^{-3} \cdot \fD_{k_3+1}$, and there are $h_3+1$ choices for $(h_1, h_2)$. So, the summation is at most $\ten^{-3} \cdot \sum_{h_3+k_3 \ge k} (h_3+1) \fD_{k_3+1} \delta^{h_3+k_3} \le \frac{4}{\ten} \cdot \sum_{h_3+k_3 \ge k} (h_3+1) \fD_{k_3} \delta^{h_3+k_3}$, where the last inequality holds because $\frac{\fD_{k_3+1}}{\fD_{k_3}} \le \ten^2 \cdot \frac{\fC_{k_3+1}}{\fC_{k_3}} \le 4 \cdot \ten^2$.

    Overall, replacing indices accordingly, we can write \eqref{eq:pk-bash} as at most
\begin{align*}
    &\hspace{0.5cm} (2\delta)^k + 2 \cdot \left(2(k+1) \cdot \delta^k + 2 \cdot \sum_{\substack{a, b \ge 0 \\ a+b \ge k \\ (a, b) \neq (0, k)}} (a+1) \fD_b \delta^{a+b} + \frac{4}{\ten} \cdot \sum_{\substack{a, b \ge 0 \\ a+b \ge k}} (a+1) \fD_b \delta^{a+b} \right) \\
    &\le (2\delta)^k + 2 \cdot \left(2(k+1) \cdot \delta^k + 3 \cdot \sum_{\substack{a, b \ge 0 \\ a+b \ge k \\ (a, b) \neq (0, k)}} (a+1) \fD_b \delta^{a+b} + \frac{4}{\ten} \fD_k \delta^k \right).
\end{align*}

    We can now focus on the middle summation term. If we first consider all terms with $b = 0$, the sum equals $\sum_{a \ge k} (a+1) \delta^a = (k+1) \delta^k + (k+2) \delta^{k+1} + \cdots \le 2 (k+1) \delta^k,$ as long as $\delta \le \frac{1}{3 \cdot \ten^3}$. For the remaining terms, we fix $d = a+b$ and consider the sum. If $d = k$, the sum equals $\delta^k \cdot (2 \fD_{k-1} + 3 \fD_{k-2} + \cdots + k \fD_1)$. Since $\fD_{n+1} \ge \ten^2 \fD_n$ for all $n \ge 1$, this is at most $\delta^k \cdot 4 \fD_{k-1}$. For $d > k$, the sum equals $\delta^d \cdot (\fD_{d} + 2 \fD_{d-1} + \cdots + d \fD_1) \le 2 \delta^d \cdot \fD_d$. Since $\fD_d \le 4 \cdot \ten^2 \cdot \fD_{d+1}$, as long as $\delta \le \frac{1}{3 \cdot \ten^3}$, the terms $2 \delta^{d} \cdot \fD_d$ decrease by a factor greater than $2$ each time $d$ increases. So the sum over all $d > k$ is at most $4 \delta^{k+1} \cdot \fD_{k+1}.$ Overall, the summation in the middle term is at most $2(k+1) \delta^k + 4 \fD_{k-1} \cdot \delta^k + 4 \fD_{k+1} \cdot \delta^{k+1}$.

    Overall, this means \eqref{eq:pk-bash} is at most
\begin{equation} \label{eq:pk-bound}
    2^k \delta^k + 16(k+1) \cdot \delta^k + 24 \fD_{k-1} \cdot \delta^k + 24 \fD_{k+1} \cdot \delta^{k+1} + \frac{8}{\ten} \fD_k \delta^k.
\end{equation}
    Now, note that $\frac{\fD_{k-1}}{\fD_k} \le \frac{1}{\ten}$ for all $k \ge 1$, even for $k = 1$. Moreover, $\frac{\fD_{k+1}}{\fD_k} \le \ten^2 \cdot \frac{\fC_{k+1}}{\fC_k} \le 4 \cdot \ten^2$. Thus, \eqref{eq:pk-bound} is at most
\[\delta^k \cdot \left(2^k + 16(k+1) + \frac{32}{100} \cdot \fD_k + 96 \cdot \ten^2 \cdot \fD_k \cdot \delta\right)\]
    Assuming that $\delta \le \frac{1}{3 \cdot \ten^3}$, this is at most $\delta^k \cdot \left(2^k + 16(k+1) + 0.64 \cdot \fD_k\right),$ which can be verified to be at most $\delta^k \cdot \fD_k$ for all $k \ge 1$, by just using the fact that $\fD_k \ge 100^k$ for all $k \ge 1$. This completes the inductive step.
\end{proof}

We are now ready to prove \Cref{lem:main-technical}.

\begin{proof}[Proof of \Cref{lem:main-technical}]
    If $f_m < m$, this means that either the event $p_1(\ba, \bb)$ occurs, or there exist indices $i < i'$ with $w_i = w_{i'} = 1$ but we fall behind at least $K+1$ steps from step $i$ to step $i'$.

    Assuming $\delta \le \frac{1}{3 \cdot 10^3}$, the probability of $p_1(\ba, \bb)$ is at most $100 \delta$. Alternatively, if there exist $i < i'$ with $w_i = w_{i'} = 1$ but we fall behind at least $K+1$ steps from step $i$ to step $i'$, there must exist such an $i, i'$ with minimal $i'-i$ (breaking ties arbitrarily). This could be because $w_{i+1} = w_{i+2} = \dots = w_{i+r} = 0$ for some $r \ge K/2$. However, the probability of there being $r \ge K/2$ consecutive indices $w_{i+1} = w_{i+2} = \dots = w_{i+r} = 0$ is at most $n \cdot \delta^{K/2} \le \delta$. 
    
    The final option is that, if we look at the first index $i+r > i$ with $w_{i+r} = 0$, $r \le K/2$. This means that from step $i+r$ to $i'$, we must fall behind at least $K/2$ steps, and there could not have been any intermediate steps where we fell behind more than $K$ steps. Hence, if we restrict ourselves to the strings $\ba_{i+r:i'}$ and $\bb_{f_{i+r}, f_{i'}}$, the event indicated by $p_k(\ba_{i+r:i'}, \bb_{f_{i+r}:})$ must occur, since conditioned on $f_{i+r}$ and the fact that $w_{i+r}=w_{i'}=1$, the value $f_{i'}$ only depends on $\ba_{i+r:i'}$, $\bb$ starting from position $f_{i+r}$, and $w_{i+r+1}, \dots, w_{i'-1}$.

    In other words, there exists some contiguous subsequences $\ba'$ and $\bb'$ of $\ba$ and $\bb$, respectively, such that the event of $p_{K/2}(\ba', \bb')$ occurs. For any fixed $\ba', \bb'$, the probability is at most $(4 \cdot 100^2 \cdot \delta)^{K/2}$. Since there are at most $n^2$ possible contiguous subsequences for each of $\ba'$ and $\bb'$, the overall probability is at most $n^4 \cdot (4 \cdot 100^2 \cdot \delta)^{K/2} \le 50 \delta$, assuming that $\delta \le \frac{1}{3 \cdot 10^6}$ and $K = C_2 \log n$ where $C_2$ is sufficiently large.

    Overall, the probability of falling behind is at most $100 \delta + \delta + 50 \delta \le 200 \delta$.
\end{proof}

\section{Full algorithm/analysis} \label{sec:algorithm}

Let us depict the true string $x \in \{0, 1\}^n$ as $\underbrace{0 \dots 0}_{a_0 \text{ times}} 1 \underbrace{0 \dots 0}_{a_1 \text{ times}} 1 \cdots 1 \underbrace{0 \dots 0}_{a_t \text{ times}},$ i.e., there are $t-1$ ones, and the string starts and ends with a run of $0$'s. This assumption can be made WLOG by padding the string with $L$ $0$'s at the front and the end. For any $L$-separated string, doing this padding maintains the $L$-separated property, and we can easily simulate the padded trace by adding $\Bin(L, 1-\delta)$ $0$'s at the front and $\Bin(L, 1-\delta)$ $0$'s at the back. Once we reconstruct the padded string, we remove the padding to get $x$.

We assume we know the value of $t$. Indeed, the number of $1$'s in a single trace $\tilde{x}$ is distributed as $\Bin(t, 1-\delta)$. So, by averaging the number of $1$'s over $O(n \log n)$ random traces and dividing by $1-\delta$, we get an estimate of $t-1$ that is accurate within $0.1$ with $1-\frac{1}{n^{10}}$ probability. Thus, by rounding, we know $t$ exactly with $1-\frac{1}{n^{10}}$ probability.

The main goal is now to learn the lengths $a_0, a_1, \dots, a_t$. If we learn these exactly just using the traces, this completes the proof. Our algorithm runs in two phases: a coarse estimation phase and a fine estimation phase. In the coarse estimation phase, we sequentially learn each $a_i$ up to error $O(\sqrt{a_i \log n})$. In the fine estimation phase, we learn each $a_i$ exactly, given the coarse estimates.

\subsection{Coarse estimation}

Fix some $0 \le m \le t$, and suppose that for all $i < m$, we have estimates $b_i$ satisfying $|b_i-(1-\delta) a_i| \le 10 \sqrt{a_i}$. (If $m = 0$, then we have no estimates yet.) Our goal will be to provide an estimate $b_m$ such that $|b_m-(1-\delta) a_m| \le 10 \sqrt{a_m}$.

Consider a trace $\tilde{x}$ of $x$. Let $w_0 = w_{t+1} = 1$ and for each $1 \le i \le t$, let $w_i$ be the indicator that the $i$th $1$ is retained. Next, for each $0 \le i \le t$, let $\tilde{a}_i \sim \Bin(a_i, 1-\delta)$ represent the number of $0$s in the $i$th run that were not deleted. Note that with at least $0.99$ probability, $|\tilde{a}_i-(1-\delta) a_i| \le 10 \sqrt{\log n \cdot a_i}$ for all $i$. Since $|b_i-(1-\delta) a_i| \le 10 \sqrt{a_i}$ for all $i < m$, this implies that $|\tilde{a}_i-b_i| \le 20 \sqrt{\log n \cdot b_i}$ for all $i < m$.

Now, even though we have no knowledge of $\tilde{a}_i$ or $a_i$, we can still simulate the probabilistic process of \Cref{sec:main-process}. Let $0 = i_0 < i_1 < \cdots < i_h = t+1$ be the list of all indices $i: 0 \le i \le t+1$ with $w_i = 1$. While we do not know the values $\tilde{a}_i$, for every pair of consecutive indices $i_q, i_{q+1}$, the value $\tilde{a}_{i_q:i_{q+1}}$ is exactly the number of $0$'s between the $q$th and $(q+1)$st $1$ in the trace $\tilde{x}$ (where we say that the $0$th $1$ is at position $0$ and the $(t+1)$st $1$ is at position $|\tilde{x}|+1$). In other words, if $r_q$ represents the position of the $q$th $1$, then $\tilde{a}_{i_q:i_{q+1}} = r_{q+1}-r_q-1$. Hence, because computing each $f_{i_{q+1}}$ only requires knowledge of $\bb$ and the value of $\tilde{a}_{i_q:i_{q+1}}$, and since $f_{i_0} = f_0 = 0$, the algorithm can in fact compute $g_q := f_{i_q}$ for all $0 \le q \le h$, using the same process as described in \Cref{sec:main-process}, even if the values $i_q$ are not known.

Algorithm~\ref{alg:findmthone} simulates this process, assuming knowledge of $m$, $b_0, \dots, b_{m-1}$, a single trace $\tilde{x}$, and $t$. In Algorithm~\ref{alg:findmthone}, we use the variable $\val$ to represent $g_q = f_{i_q}$, i.e., the current prediction of the position $i_q$. In other words, $\val - \, i_q$ equals the number of steps ahead (or $i_q - \val$ equals the number of steps behind) we are. 

\begin{algorithm}
\caption{Locate the $m$th and $(m+1)$st $1$ in $x$, in the trace $\tilde{x}$, and return the position and length of the gap.}
\label{alg:findmthone}
\begin{algorithmic}[1]
\Procedure{Align}{$\tilde{x}, t, m, b_0, \dots, b_{m-1}$}
\State Let $r_q$ be the position of the $q$th $1$ in $\tilde{x}$, for each $1 \le q \le t-1$.
\State $r_0 \leftarrow 0$, $r_t \leftarrow |\tilde{x}|+1$.
\State $\val \leftarrow 0$, $q \leftarrow 0$
\While{$\val < m$}
    \State Find the smallest $j'$ such that $\exists j, j'$ with $\val \le j < j'$ and $|(r_{q+1}-r_{q}-1) - b_{j:j'}| \le C_0 \log n \cdot \sqrt{b_{j:j'}}.$
    \If{no such $j, j'$ exist}
        \State \textbf{Return FAIL}
    \EndIf
    \State $\val \leftarrow j'$
    \State $q \leftarrow q + 1$
    \EndWhile
\If{$\val = m$}
    \State \textbf{Return} $(q, r_{q+1}-r_{q}-1)$.
\Else
    \State \textbf{Return FAIL}
\EndIf
\EndProcedure
\end{algorithmic}
\end{algorithm}

\begin{lemma} \label{lem:crude-main-analysis}
    Fix $b_0, \dots, b_{m-1}$ such that $|b_i-(1-\delta) a_i| \le 10 \sqrt{a_i}$ for all $0 \le i \le m-1$. With probability at least $0.98$ over the randomness of $\tilde{x}$, we have that \Cref{alg:findmthone} returns $q$ such that the $q$th $1$ in $\tilde{x}$ corresponds to the $m$th $1$ in $x$. Moreover, conditioned on this event holding, the distribution $r_{q+1}-r_q-1$ exactly follows $\Bin(a_m, 1-\delta)$.
\end{lemma}

\begin{proof}
    Let us first condition on the values $\tilde{a}_0, \dots, \tilde{a}_{m-1}$, assuming that $|\tilde{a}_i-(1-\delta)a_i| \le 10 \sqrt{\log n \cdot a_i}$ for all $0 \le i \le m-1$. As discussed earlier, this occurs with at least $0.99$ probability, and implies that $|\tilde{a}_i-b_i| \le 20 \sqrt{\log n \cdot b_i}$ for all $i < m$.

    Let us also condition on $w_m = 1$. By \Cref{lem:not-ahead} and \Cref{lem:main-technical}, the probability that $f_m = m$, for $\delta = \frac{1}{3 \cdot 10^6}$, is at least $0.99$. This is conditioned on $w_m = 1$ and the values $\tilde{a}_1, \dots, \tilde{a}_{m-1}$ (assuming $|\tilde{a}_i-b_i| \le 20 \sqrt{\log n \cdot b_i}$). This means that with at least $0.99$ probability, the algorithm finds the position $q$ with $i_{q} = m$.
    Since $f_m$ only depends on $\bb$, $\tilde{\ba}_{0:m}$ and $w_1, \dots, w_m$, with probability at least $0.99 \cdot (1-\delta) \cdot 0.99$ over the randomness of $w_1, \dots, w_m$ and $\tilde{a}_1, \dots, \tilde{a}_{m-1}$, we have that $w_m = 1$ and $i_{q} = m$. This is independent of $w_{m+1}$, so with probability at least $0.99^2 \cdot (1-\delta)^2 \ge 0.98$ probability, we additionally have that $w_{m+1} = 1$.
    
    The event that $i_q = m$ means that $r_q$ is the position in $\tilde{x}$ of the $m$th $1$ in the true string $x$. Moreover, since neither the $m$th nor $(m+1)$th $1$ was deleted, $r_{q+1}$ is the position in $\tilde{x}$ of the $(m+1)$th $1$ in the true string $x$. So, $r_{q+1}-r_q-1$ is in fact the length of the gap between the $m$th and $(m+1)$th $1$ after deletion, which means it has length $\tilde{a}_m \sim \Bin(a_m, 1-\delta)$, since $\tilde{a}_m$ is independent of the events that decide whether $w_m = w_{m+1} = 1$ and $i_q = m$.
\end{proof}

Given this, we can crudely estimate every gap, in order. Namely, assuming that that we have estimates $b_0, \dots, b_{m-1}$ (where $0 \le m \le t$), we can run the \textsc{Align} procedure on $O(\log n)$ independent traces. By a Chernoff bound, with $\frac{1}{n^{15}}$ failure probability, at least $0.9$ fraction of the traces will have the desired property of \Cref{lem:crude-main-analysis}, so will output some $(q, b)$ where $b \sim \Bin(a_m, 1-\delta)$. Since $\Bin(a_m, 1-\delta)$ is in the range $a_m(1-\delta) \pm 10 \sqrt{a_m}$ with at least $0.99$ probability, at least $0.75$ fraction of the outputs $(q, b)$ will satisfy $|b - (1-\delta) a_m| \le 10 \sqrt{a_m}$, with $\frac{1}{n^{15}}$ failure probability. Thus, by defining $b_m$ to be the median value of $b$ across the randomly drawn traces, we have that $|b_m-(1-\delta) a_m| \le 10 \sqrt{a_m}$ with at least $1 - \frac{1}{n^{10}}$ probability.

By running this procedure iteratively to provide estimates $b_0, b_1, \dots, b_{t}$, we obtain Algorithm~\ref{alg:crude}. The analysis in the above paragraph implies the following result.

\begin{theorem}[Crude Approximation] \label{thm:crude}
    Algorithm~\ref{alg:crude} uses $O(n \log n)$ traces and polynomial time, and learns estimates $b_0, b_1, \dots, b_{t}$ such that with at least $1 - \frac{1}{n^9}$ probability, $|b_m-(1-\delta) a_m| \le 10 \sqrt{a_m}$ for all $0 \le m \le t$.
\end{theorem}

\begin{algorithm}
\caption{Crude Estimation of all gaps}
\label{alg:crude}
\begin{algorithmic}[1]
\Procedure{Crude}{}
\State Use $O(n \log n)$ traces to compute $t$, where $t$ equals the number of $0$s in $x$.
\For{$m = 0$ to $t$}
    \For{$i = 1$ to $O(\log n)$}
        \State Draw trace $\tilde{x}^{(i)}$.
        \State $(q^{(i)}, b^{(i)}) \leftarrow \textsc{Align}(\tilde{x}^{(i)}, t, m, b_0, \dots, b_{m-1})$
    \EndFor
    \State Let $b_m$ be the median of $b^{(1)}, \dots, b^{(O(\log n))}$ \Comment{Some of the outputs $(q^{(i)}, b^{(i)})$ may be \textbf{FAIL}, but we can let $b^{(i)}$ be an arbitrary real number if \textsc{Align} failed on $\tilde{x}^{(i)}$, so that the median is well-defined.}
\EndFor
\State \textbf{Return} $(b_0, b_1, \dots, b_{t})$
\EndProcedure
\end{algorithmic}
\end{algorithm}

\subsection{Fine estimation}

In this section, we show how to exactly compute each $a_m$ with high probability, given the crude estimates $b_0, b_1, \dots, b_{t-1}$. This will again be done using an alignment procedure, but this time running the alignment both ``forward and backward''.

Namely, given a trace $\tilde{x}$, we will try to identify the $m$th and $(m+1)$st $1$'s from the original string, but we try to identify the $m$th $1$ by running \textsc{Align} on $\tilde{x}$ and the $(m+1)$st $1$ by running \textsc{Align} on the reverse string $\rev(\tilde{x}) := \tilde{x}_{|\tilde{x}|} \cdots \tilde{x}_2 \tilde{x}_1$. The idea is: assuming that we never go ahead in the alignment procedure, if we find some index $q$ in the forward alignment procedure with $g_q = f_{i_q} = m$, then the true position $i_q$ must be at least $m$. Likewise, if we do the alignment procedure in reverse until we believe we have found the $(t-m)$th $1$ from the back (equivalently, the $(m+1)$th $1$ from the front), the true position must be at most $m+1$.

So, the true positions of the index found in the forward alignment procedure can only be earlier than that of the index from the backward alignment procedure, if the true positions were exactly $m$ and $m+1$, respectively. Thus, by comparing the indices, we can effectively verify that the positions are correct, with negligible failure probability (rather than with $1-O(\delta)$ failure probability). This is the key towards obtaining the fine estimate of $a_m$, rather than just a coarse estimate that may be off by $O(\sqrt{a_m})$. 

Algorithm~\ref{alg:fine} formally describes the fine alignment procedure, using $N = O(n \log n)$ traces, assuming we have already done the coarse estimation to find $b_0, b_1, \dots, b_t$. 

\begin{algorithm}
\caption{Fine Estimation of all gaps}
\label{alg:fine}
\begin{algorithmic}[1]
\Procedure{Fine}{$t, b_0, \dots, b_t$}
\State Draw $N = O(n \log n)$ traces $\tilde{x}^{(1)},\dots,\tilde{x}^{(N)}$.
\For{$m = 0$ to $t$}
    \State Initialize $b^{(1)}, b^{(2)}, \dots, b^{(N)} \leftarrow \textbf{NULL}$.
    \For{$i = 1$ to $N$}
        \State $\tilde{m} \leftarrow$ number of $1$'s in $\tilde{x}$.
        \State $(q_{\text{f}}, b_{\text{f}}) \leftarrow \textsc{Align}(\tilde{x}^{(i)}, t, m, b_0, b_1 \dots, b_t)$.
        \State $(q_{\text{b}}, b_{\text{b}}) \leftarrow \textsc{Align}(\rev(\tilde{x}^{(i)}), t, t-m, b_t, b_{t-1}, \dots, b_0)$.
        \If{$q_{\text{f}} + q_{\text{b}} = \tilde{m}$}
            \State $b^{(i)} \leftarrow b_{\text{f}}$
        \EndIf
        \State Set $a_m$ to be $\frac{1}{1-\delta}$ times the average of all non-null $b^{(i)}$'s, rounded to the nearest integer.
    \EndFor
\EndFor
\State \textbf{Return} $(a_0, a_1, \dots, a_{t})$
\EndProcedure
\end{algorithmic}
\end{algorithm}

\begin{lemma} \label{lem:word-pushing}
    Suppose that $|b_i-(1-\delta) a_i| \le 10 \sqrt{a_i}$ for all $1 \le 0 \le t$. Fix indices $0 \le m \le t$ and $1 \le i \le N$, and for simplicity of notation, let $\tilde{x} := \tilde{x}^{(i)}$. Let $\tilde{m}$ be the number of $1$'s in $\tilde{x}$. Then, the probability that $q_{\text{f}}+q_{\text{b}} = \tilde{m}$, but either the forward or backward iterations finds an index in $\tilde{x}$ which does not correspond to the $m$th $1$ or $(m+1)$th $1$, respectively, from $x$, is at most $2n^{-10}$. Moreover, if the forward and backward iterations find indices in $\tilde{x}$ corresponding to the $m$th $1$ and $(m+1)$th $1$, respectively, then $q_{\text{f}}+q_{\text{b}} = \tilde{m}$. Finally, the probability of finding both corresponding indices is at least $0.98$.
\end{lemma}

\begin{proof}
    First, let us consider the forward alignment procedure. We know that $\val$ tracks $f_{i_q}$ when looking at the $q$th $1$ of $\tilde{x}$ (from left to right). So, if we do not return \textbf{FAIL}, then $f_{i_{q_{\text{f}}}} = m$. If $i_{q_{\text{f}}} < m$, this implies there is an index $i = i_{q_{\text{f}}}$ where $f_i > i$. The probability of this is at most $n^{-10}$, by \Cref{lem:not-ahead}. Otherwise, $i_{q_{\text{f}}} \ge m$, meaning that the $q_{\text{f}}$th $1$ in $\tilde{x}$ is after (or equal to) the $m$th $1$ in $x$.

    Likewise, if we consider the backward alignment procedure, if we do not return \textbf{FAIL}, then except for an event with probability at most $n^{-10}$, the $q_{\text{b}}$th $1$ in $\rev(\tilde{x})$ is ahead of (or equal to) the $(t-m)$th $1$ in $\rev(x)$. Equivalently, the $(\tilde{m}+1-q_{\text{b}})$th $1$ in $\tilde{x}$ (reading from left to right) is before (or equal to) the $(m+1)$th $1$ in $x$ (reading from left to right).

    So, barring a $2 \cdot n^{-10}$ probability event, the only way that the $q_{\text{f}}$th $1$ in $\tilde{x}$ is strictly before the $(\tilde{m}+1-q_{\text{b}})$th $1$ in $\tilde{x}$ is if the $q_{\text{f}}$th $1$ in $\tilde{x}$ is precisely the $m$th $1$ in $x$ and $(\tilde{m}+1-q_{\text{b}})$th $1$ in $\tilde{x}$ is precisely the $(m+1)$th $1$ in $x$. However, if $q_{\text{f}}+q_{\text{b}} = \tilde{m}$, then in fact the $q_{\text{f}}$th $1$ is before the $(\tilde{m}+1-q_{\text{b}})$th $1$ in $\tilde{x}$ (reading from left to right). This proves the first statement.

    Next, if we in fact found the corresponding indices, they are consecutive $1$'s in $x$, which means they must be consecutive $1$'s in $\tilde{x}$. So, if we found the $q_{\text{f}}$th $1$ from the left, and the $q_{\text{b}}$th $1$ from the right, we must have $q_{\text{f}}+q_{\text{b}} = \tilde{m}$.

    Finally, the event of finding both corresponding indices is equivalent to $f_m = m$ in the forward iteration and $f_{t-m} = t-m$ in the backward iteration. Conditioned on the corresponding $1$'s \emph{not} being deleted, each of these occur with at least $0.98$ probability, by Lemmas~\ref{lem:not-ahead} and~\ref{lem:main-technical}. So, the overall probability is at least $0.9$.
\end{proof}

We are now ready to prove \Cref{thm:main}. Indeed, given the accuracy of the crude estimation procedure, it suffices to check that for each $m$, we compute $a_m$ correctly, with at least $1 - n^{-5}$ probability.

\begin{theorem}[Fine Estimation] \label{thm:fine}
    Assume that $t$, the number of ones in $x$, is computed correctly, and for all $0 \le m \le t$, $|b_m-(1-\delta) a_m| \le 10 \sqrt{a_m}$.

    Then, for any fixed $m: 0 \le m \le t$, with at least $1 - n^{-5}$ probability, we compute the gap $a_m$ correctly.
\end{theorem}

\begin{proof}
    For any fixed iteration $i: 1 \le i \le N$, if both the forward and backward procedures correctly identify the $m$th and $(m+1)$th $1$'s from the left, respectively, then $q_{\text{f}}+q_{\text{b}} = \tilde{m}$ by \Cref{lem:word-pushing}. In this case, we will compute an actual value $b^{(i)} = b_{\text{f}}$. Moreover, as discussed in the proof of \Cref{lem:crude-main-analysis}, the event that the forward procedure correctly identifies the right $1$ only depends on $\bb$, $\hat{a}_0, \dots, \hat{a}_{m-1}$, and the events of whether the first $m$ $1$'s are deleted. Thus, the event that the backward procedure correctly identifies the right $1$ only depends on $\bb$, $\hat{a}_{m+1}, \dots, \hat{a}_t$, and the events of whether the $(m+1)$th $1$ until the $t$th $1$ are deleted.

    Thus, the forward and backward procedure correctly identifying the right $1$'s is independent of $\hat{a}_m \sim \Bin(a_m, 1-\delta)$. Moreover, in this case, $b_{\text{f}}$ is precisely $\hat{a}_m$, since $q_{\text{f}}$ is the position in $\tilde{x}$ corresponding to the $m$th $1$ in $x$, and neither the $m$th nor $(m+1)$th $1$ can be deleted if both of these $1$'s are identified.

    So, if the forward and backward procedures identifying the right $1$'s for trace $\tilde{x}^{(i)}$, the conditional distribution of $b^{(i)}$ is $\Bin(a_m, 1-\delta)$. However, we really want to look at the distribution conditioned on the event $q_{\text{f}}+q_{\text{b}} = \tilde{m}$. Indeed, by \Cref{lem:word-pushing}, this event is equivalent to either the forward and backward procedures identifying the right $1$'s, or some other event which occurs with at most $2n^{-10}$ probability. Because $b^{(i)}$ is clearly between $0$ and $n$, and since the probability of both $1$'s being correctly identified is at least $0.9$ by \Cref{lem:word-pushing}, the expectation of $b^{(i)}$, conditioned on not being \textbf{NULL}, is $a_m(1-\delta) \pm O(n^{-10} \cdot n) = a_m(1-\delta) \pm O(n^{-9})$.

    By a Chernoff bound, the number of $1 \le i \le N$ with $b^{(i)} \neq \textbf{NULL}$ is at least $0.5 \cdot N$ with at least $1 - n^{-10}$ probability, since in expectation it is at least $0.9 N$. Then, by another Chernoff bound, the empirical average of all such $b^{(i)}$ is within $0.1$ of its expectation with $1 - n^{-10}$ probability, which is $a_m(1-\delta) \pm O(n^{-9})$. Thus, taking the empirical average and dividing by $1-\delta$, with at most $O(n^{-10})$ failure probability, $\frac{1}{1-\delta}$ times the average of all non-null $b^{(i)}$'s is within $0.2$ of $a_m$, and thus rounds to $a_m$.
\end{proof}

\bibliographystyle{alpha}
\bibliography{biblio}

\end{document}